\documentclass[a4paper,onecolumn,11pt,accepted=2023-07-31]{quantumarticle}
\pdfoutput=1
\usepackage[utf8]{inputenc}
\usepackage[english]{babel}
\usepackage[T1]{fontenc}
\usepackage{amsmath}
\usepackage{hyperref}
\usepackage{tikz}
\usepackage{lipsum}
\usepackage{amsthm}
\usepackage{mathtools}

\usepackage[normalem]{ulem}

\newtheorem{theorem}{Theorem}
\newtheorem{corollary}{Corollary}
\newtheorem{lemma}{Lemma}
\newtheorem{proposition}{Proposition}

\def\mc{\mathcal}

\begin{document}

\title{A Converse for Fault-tolerant Quantum Computation}%


\author{Uthirakalyani~G$^\dagger$}
\thanks{ $^\dagger$ Equally contributing student authors.}
\affiliation{Department of Electrical Engineering, Indian Institute of Technology Madras, Chennai, India.}
\email{ee19d404@smail.iitm.ac.in}

\author{Anuj~K.~Nayak$^\dagger$}
\affiliation{Department of Electrical and Computer Engineering, University of Illinois at Urbana-Champaign, Urbana, USA.}
\email{anujk4@illinois.edu}

\author{Avhishek~Chatterjee}
\affiliation{Department of Electrical Engineering, Indian Institute of Technology Madras, Chennai, India.}
\email{avhishek@ee.iitm.ac.in.}


\maketitle

\begin{abstract}
As techniques for fault-tolerant quantum computation keep improving, it is natural to ask: what is the fundamental lower bound on space overhead? In this paper, we obtain a lower bound on the space overhead required for $\epsilon$-accurate implementation of a large class of operations that includes unitary operators. For the practically relevant case of sub-exponential  depth and sub-linear gate size, our bound on space overhead is tighter than the known lower bounds.  We obtain this bound by connecting fault-tolerant computation with a set of finite blocklength quantum communication problems whose accuracy requirements   satisfy a joint constraint. The lower bound on space overhead obtained here leads to a strictly smaller upper bound on the noise threshold for noise that are not degradable. Our bound directly extends to the case where noise at the outputs of a gate are non-i.i.d. but noise across gates are i.i.d. 
\end{abstract}


\section{Introduction}
\label{sec:intro}
The idea of a quantum computer was proposed in \cite{Benioff1980}. The computational advantage of a quantum computer over its classical counterpart was first shown mathematically in \cite{Deutsch1985}, followed by \cite{DeutschJ1992}. However, this initial excitement for quantum  computers was confronted with a practical issue, noise in quantum circuits. 

To tackle noise in quantum circuits, the exciting area of quantum fault-tolerance emerged, thanks to the work in \cite{Shor1996}, followed by \cite{Steane1996,AharonovB1997,Kitaev1997} and many others. The seminal works in \cite{Shor1996,Steane1996,AharonovB1997,Kitaev1997} essentially showed that if the strength of the noise at the gates is below a threshold, almost accurate quantum computation can be realized. However, this is achieved at the cost of a poly-logarithmic increase in the size of the circuit with respect to the ideal circuit made of noiseless gates. In subsequent works, this poly-logarithmic space overhead has been improved. In \cite{Gottesman2014,Fawzi2020}, it was shown that a constant space overhead can be achieved if the noise strength is below a threshold. This naturally raises the question: what is the minimum space overhead requirement? 

Relatively fewer works have addressed the question of minimum space overhead requirement. However, as better and better codes are found, the interest in understanding the fundamental limit on space overhead is growing \cite{Razborov2004,KempeRUW2008,FawziMS2022}. 

In this paper, we obtain a lower bound on the required space overhead for a broad class of noise models. In the  practically likely regime of sub-exponential (in the number of input qubits) depth and sub-linear gate size, \cite{Gottesman2014,Fawzi2020}, our bound is strictly better than the existing bounds in \cite{HarrowN2003, Razborov2004,KempeRUW2008,FawziMS2022}. This bound is obtained by connecting the fault-tolerant computation problem to a set of finite blocklength communication problems whose accuracy requirements satisfy a joint constraint.

\subsection{Related Work}
\label{sec:contribution}
Influenced by the gradual improvements  in the space overhead of the error correcting schemes for quantum circuits \cite{Steane1996,AharonovB1997,Gottesman2014,Fawzi2020}, improved lower bounds on space overhead were sought.

Harrow and Nielsen obtained a threshold of $0.74$ for depolarizing noise \cite{HarrowN2003}. In \cite{Razborov2004}, Razborov obtained an improved gate size dependent threshold
$1 - \frac{1}{g}$, where $g$ is the gate size, i.e., the maximum number of inputs to a gate. Kempe et al. \cite{KempeRUW2008} improved this bound to $1 - \sqrt{2^{\frac{1}{g}}-1}$ for mixture of unitary gates of size $g$.  

Buhrman et al. \cite{BuhrmanCLLSU2006} and Virmani et al. \cite{VirmaniHP2005} showed classical simulability of noisy quantum circuits beyond a threshold under assumptions on special gate operations and noise. These indirectly showed that, under certain assumptions, a quantum computer loses its edge over a classical computer if the noise is more than a threshold.

Recently, Fawzi et al. \cite{FawziMS2022} obtained a threshold for quantum computation in terms of the quantum capacity of a channel with the same noise. Their model includes any arbitrary gate operation and allows free noiseless classical computation. Though their threshold does not depend on the gate size $g$, it is strictly better than the threshold in \cite{KempeRUW2008} for $g \ge 2$. Moreover, they provide a lower bound on the space overhead when the noise is below the threshold. 


\subsection{Our Contribution}
The main result in this paper is a lower bound on the number of physical qubits needed for fault-tolerant  quantum computation with a quantum state of $d$ (logical) qubits. The bound is obtained by first establishing a connection between fault-tolerant computing  and  a set of  finite blocklength quantum communication problems \cite{KhatriW2020book} whose accuracy requirements share a joint constraint. Then, using results from finite blocklength quantum  communication \cite{KhatriW2020book} we obtain a converse, which is then optimized to obtain the final lower bound. Our approach is inspired by information theoretic converse for classical noisy circuits in \cite{Pippenger1988, EvansS1999}, but the techniques used are quite different.

In the likely practical scenario of sub-exponential circuit depth and constant gate size, our bound is tighter than the best known bound from \cite{FawziMS2022}. Specifically, if the input has $d$ qubits and the gate noise is given by the quantum channel $\mc{N}$, then \cite{FawziMS2022} shows that the number of physical qubits needed is $\frac{d}{Q(\mc{N})}$, where $Q(\mc{N})$ is the quantum capacity of the channel $\mc{N}$. We show that for sub-linear gate size, i.e., $g=o(d)$, the lower bound
is $\frac{d~g}{I_c(\mc{N}^{\otimes g})}$, where $I_c(\cdot)$ denotes the coherent information of the corresponding quantum channel. Since $Q(\mc{N}) = \sup_{k} \frac{I_c(\mc{N}^{\otimes k})}{k}$, our bound is tighter. 

An implication of our bound is an improved upper bound on noise threshold for a broad class of noise and gate models. For noise that are not degradable, this upper bound is strictly lower than the known ones and depends on gate size $g$ like \cite{Razborov2004,KempeRUW2008}. Interestingly, in the case where noise at a gate are correlated but noise across gates are i.i.d., bounds on space overhead and noise threshold can be obtained by replacing  $I_c(\mc{N}^{\otimes g})$ with $I_c(\mc{N}^{(g)})$ in our bound. Here, $\mc{N}^{(g)}$ is an arbitrary non-i.i.d. noise at the output of a gate of size $g$.

\subsection{Organization}
In Sec.~\ref{sec:model}, the quantum computation model is presented. The main result is presented in Sec.~\ref{sec:result} followed by a discussion about its implications. Proof of the main result is presented in Sec.~\ref{sec:proofLB}, which draws on a few intermediate results, whose proofs are  presented in Appendix~\ref{sec:lipschitzLB} and \ref{sec:lipschitzLemma}. We conclude in Sec.~\ref{sec:conclusion} with a short discussion on interesting directions for further explorations.

\section{Quantum Computation Model}
\label{sec:model}

We consider a fault-tolerant quantum circuit whose goal is to compute a function $f(\cdot)$ on a $d$-qubit input, i.e., a $2^{d}$ dimensional input. At the input to the circuit, there are total $N$ physical qubits on which quantum operations are performed sequentially in layers (Fig.~\ref{fig:computationModel}). The initial input, $\rho^{(d)}$, is placed on $d$ physical qubits and the remaining $N-d$ physical qubits are ancillas, which can be used for error correction. Thus, the initial state of the $N$ physical qubits is given by $\rho^{(N)}=\rho^{(d)} \otimes \left(\otimes_{a=1}^{N-d}\sigma_a \right)$, where $\sigma_a$ are ancilla qubits.

\begin{center}
\begin{figure}[h]
    \includegraphics[scale=0.5]{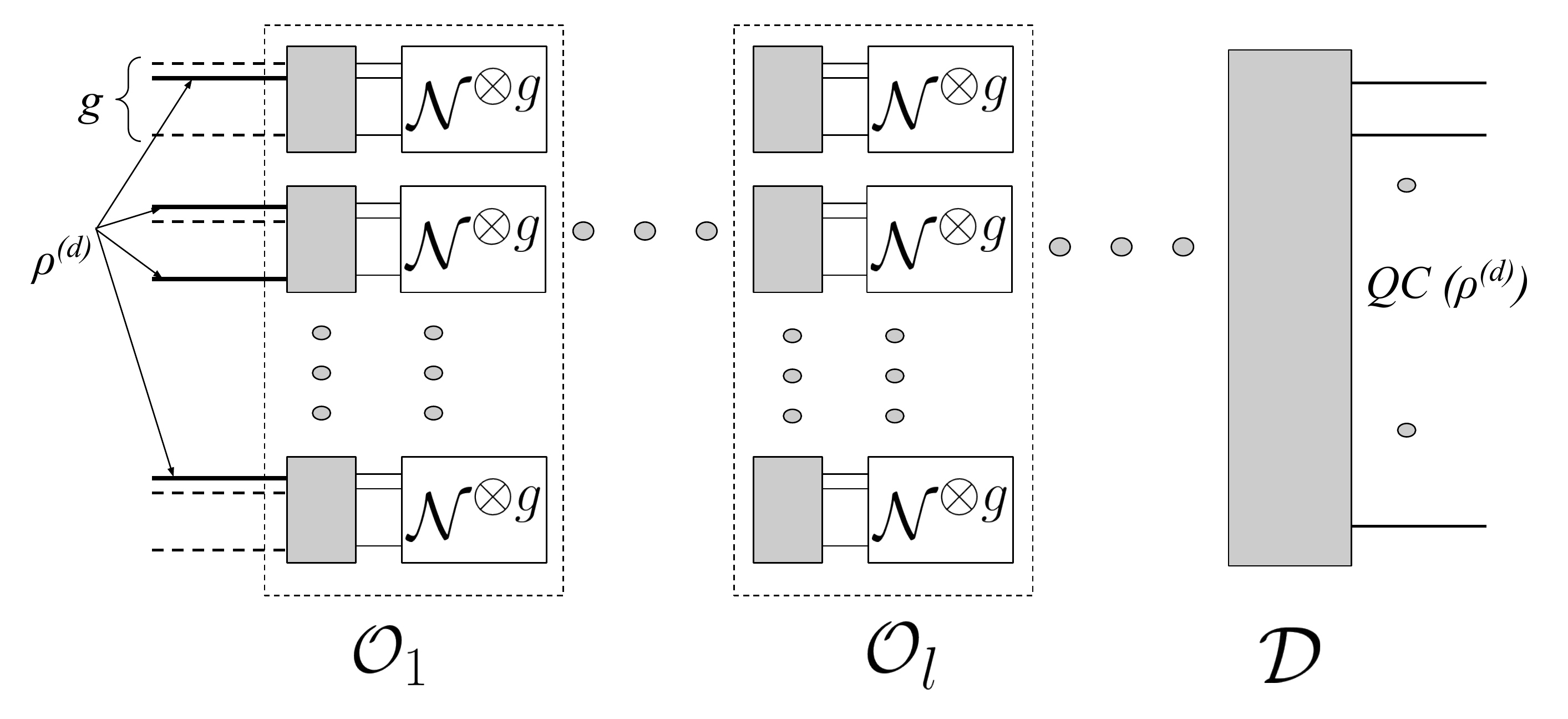}
    \caption{A schematic of the quantum computation model.}
    \label{fig:computationModel}
\end{figure}
\end{center}

 A quantum circuit of depth $D$ has $D$ layers and does $D$ operations on $\rho^{(N)}$ in a sequential manner. That is, given the quantum operations $\mathcal{O}_l$ corresponding layers $1 \le l \le D$, the output at the end of the final layer $D$ is $\mc{O}_D \circ \mc{O}_{D-1} \cdots \circ \mc{O}_1 \left(\rho^{(N)}\right)$.

The operation $\mc{O}_l$ in layer $l$ is realized using noisy quantum gates of size $g$. If a set of qubits are not processed by any gate at a particular layer, we model them to be processed by a dummy $g$-input identity gate. Thus, in the first layer, there are at least $\lceil\frac{N}{g}\rceil$ gates, including the dummy identity gates. Like \cite{Razborov2004,FawziMS2022}, at each layer, introduction of fresh ancilla qubits and  noise free classical computations are allowed.

Following prior work \cite{Razborov2004,KempeRUW2008,FawziMS2022}, a noisy gate is modeled by a perfect gate followed by $g$ i.i.d. noisy quantum channels $\mc{N}$, given by $\mc{N}^{\otimes g}$. Noise across gates and layers are assumed to be i.i.d.

Finally, an arbitrary noiseless quantum operation $\mc{D}$ is allowed after the final layer. It maps $\mc{O}_D \circ \mc{O}_{D-1} \cdots \circ \mc{O}_1 \left(\rho^{(N)}\right)$ to a quantum state of the same dimension as that of $f(\rho^{(d)})$. This operation is equivalent to the {\em noiseless} decoder in \cite{FawziMS2022} and the partial trace operation in \cite{Razborov2004}. We refer to the final output of the noisy circuit after the operation $\mc{D}$ as $QC(\rho^{(d)})$.

\subsection{Criteria for fault-tolerance}
In the above computational model, the criteria for $\epsilon$-accurate (or fault-tolerance) computation of $f(\cdot)$ using $QC(\cdot)$ is
\begin{align}
\mbox{C0: }\forall \rho^{(d)}~\mbox{(}d\mbox{-qubit state )}, F(QC(\rho^{(d)}),f(\rho^{(d)})) > 1-\epsilon,
\label{eq:cri0}
\end{align}
where, $F(\cdot,\cdot)$ is the standard fidelity \cite{KhatriW2020book}.

The usual arguments for achieving fault-tolerance against sufficiently small but constant noise using concatenated codes \cite[Ch. 10.6.1]{NielsenC2010} imply that criteria C0 can be achieved. In particular, it can be achieved using concatenated codes and poly-logarithmic space overhead for any sufficiently small $\epsilon>0$, where for any state $\sigma$ and a completely positive trace preserving map $\mc{E}$, $\mc{N}(\sigma)=(1-\epsilon) \sigma + \epsilon~\mc{E}(\sigma)$. Space overhead may be improved using better fault tolerance schemes like \cite{Gottesman2014,Fawzi2020}. However, the focus of this work is on converse results: a lower bound on required space overhead and an upper bound on noise threshold. Our computational model, fault-tolerance criteria and objective are similar to that in \cite{FawziMS2022,Razborov2004}.

\section{Converse for $\epsilon$-accurate Computing}
\label{sec:result}
The following proposition gives a lower bound on the number of physical qubits needed for fault-tolerant computation.
\begin{proposition}
\label{prop:LB}
There exists a class of functions on $d$-qubit states such that the number of physical qubits needed for implementing $\epsilon$-accurate circuits for computing these functions is lower bounded by 
\begin{align}
\frac{d}{\frac{I_c(\mc{N}^{\otimes g})}{g}+\frac{1}{2(d-g)~{\ln 2}}\left(\ln\frac{4d}{g}+\frac{8}{7}\right)}  - 2g, \nonumber
\end{align}
for $\epsilon\in (0,0.11)$ and $d \ge 2g$. Here, $I_c(\mc{N}^{\otimes g})$ is the coherent information of the channel $\mc{N}^{\otimes g}$.
\end{proposition}

\begin{proof}
    { The proof of Prop.~\ref{prop:LB} is given in Sec.~\ref{sec:proofLB}.}
\end{proof}
First, we would like to mention that the class of functions in Proposition~\ref{prop:LB} includes all unitary transformations. Hence, this class of functions can simulate evolutions of physical systems and implement  important quantum computation modules like quantum Fourier transform. 

The bound in Proposition~\ref{prop:LB} directly extends to circuits where noise at the outputs of a gate are not independent, but noise are i.i.d. across gates. In that case, instead of $I_c(\mc{N}^{\otimes g})$, the bound would have $I_c(\mc{N}^{(g)})$, where $\mc{N}^{(g)}$ represents the potentially correlated noise acting on the $g$-qubit output of a gate. This would be evident from the proof of Proposition~\ref{prop:LB}.

\subsection{Lower bound on space overhead}
Next, we consider circuits where $D$ is sub-exponential in $d$ and the gate size $g=o(d)$ since the probable practical implementations lie in this regime. The following corollary, which follows directly from Proposition~\ref{prop:LB}, gives a lower bound on the minimum required space overhead in this regime.

\begin{corollary}
\label{cor:LB}
For $g=o(d)$  and $\epsilon\in (0,0.11)$, there exist a class of functions that include unitary transformations,  such that the limiting space overhead, $\lim_{d\to \infty} \frac{N}{d}$, required for $\epsilon$-accurate implementation of those functions is lower bounded by $\frac{g}{I_c(\mc{N}^{\otimes g})}$.
\end{corollary}

When $D$ is sub-exponential and $g=o(d)$, the best known lower bound on limiting space overhead is $\inf_{k\ge 1} \frac{k}{I_c(\mc{N}^{\otimes k})}$ \cite{FawziMS2022}. Note that this is equal to the bound in Corollary~\ref{cor:LB} when $\mc{N}$ is degradable \cite[\S 3.25]{KhatriW2020book}. However, for depolarizing noise, and in general, for noise that are not degradable, the bound in Corollary~\ref{cor:LB} is strictly higher.

Another interesting aspect of the bound in Corollary~\ref{cor:LB} is that it captures the effect of the gate size on minimum space overhead, which is not captured by the existing best bound $\inf_{k\ge 1} \frac{k}{I_c(\mc{N}^{\otimes k})}$.

\subsection{Upper bound on noise threshold}

Noise threshold for a  noise model parameterized by a single parameter (e.g., depolarizing) is defined as the strength of noise beyond which quantum computation is not possible. Its dependence on the gate size was captured in \cite{Razborov2004} and \cite{KempeRUW2008}. On the other hand, in \cite{FawziMS2022}, a tighter bound involving only quantum capacity was given. Here, we obtain a tighter upper bound on noise threshold in terms of both gate size and a quantity closely related to quantum capacity.

We state results for generic noise that may have multiple parameters and characterize the parameter region where fault-tolerant computation is not possible using reasonable space overhead. The following is a direct corollary of Proposition~\ref{prop:LB}.

\begin{corollary}
\label{cor:threshold}
For $g=o(d)$  and $\epsilon\in (0,0.11)$, there exist a class of functions that include unitary transformations, such that for any parameter region of the noise $\mc{N}$ where $I_c(\mc{N}^{\otimes g})=0$, $\epsilon$-accurate computation requires $\frac{N}{d}=\Omega(\frac{d}{\ln d})$, i.e, {sub-linear space overhead} (upto a logarithmic factor) is necessary.
\end{corollary}

This corollary of Proposition~\ref{prop:LB}, unlike the threshold results in \cite{FawziMS2022, Razborov2004, KempeRUW2008}, does not give an impossibility result. Despite that it says something quite useful: constant  or poly-logarithmic space overhead, which are the gold standards for fault-tolerant schemes,  cannot be achieved when $I_c(\mc{N}^{\otimes g})=0$. 
In comparison, the best existing result says that linear or poly-logarithmic space overhead is not possible when $\sup_{k \ge 1} \frac{I_c(\mc{N}^{\otimes k})}{k}=0$ { \cite{FawziMS2022}}. 

Thus,  for noise that are not degradable, Corollary~\ref{cor:threshold} gives a strictly larger parameter region where no scheme can achieve poly-logarithmic space overhead.  An immediate implication of Corollary~\ref{cor:threshold} is an upper bound on noise threshold for depolarizing noise which is strictly lower than the best existing upper bound.

Next, we derive a much stronger phase transition of the required space overhead for the same parameter region as given by Corollary \ref{cor:threshold}. However, this result does not directly follow from Proposition~\ref{prop:LB} and is derived using intermediate results from the proof of Proposition~\ref{prop:LB}.

\begin{proposition}
\label{prop:threshold}
For $g=o(d)$  and $\epsilon\in (0,0.11)$, there exist a class of functions that include unitary transformations, such that for any parameter region of the noise $\mc{N}$ where $I_c(\mc{N}^{\otimes g})=0$, $\epsilon$-accurate computation for $\epsilon<0.11$ is not possible for $N=o(\exp(d))$.
\end{proposition}

Proof of this proposition is presented in Appendix~\ref{sec:thresholdProof}. This result also extends to noise that are correlated at a gate but are independent across gates.

For all practical purposes, Proposition~\ref{prop:threshold} is an impossibility result since it proves that in the parameter region where  $I_c(\mc{N}^{\otimes g})=0$, accurate computation requires at least exponential space overhead.  An even stronger impossibility result is discussed in Appendix~\ref{sec:sandwichedThreshold}.


In summary, Proposition~\ref{prop:threshold}  gives an upper bound on the noise threshold and Proposition~\ref{prop:LB} gives a lower on bound the minimum space overhead when the noise strength is below  that threshold. These bounds are applicable to a broad class of noise and gate models, and for circuits with sub-exponential depth, these are tighter than the respective existing bounds.

\subsection{Discussions}
The gate size dependent bounds on noise threshold and minimum space overhead in Proposition.~\ref{prop:threshold} and \ref{prop:LB}, respectively, can be useful in choosing the right experimental implementations. Consider the scenario where there are multiple options for experimental implementations using different approaches. These possible implementations have different gate sizes and encounter different types of noise. In this scenario, the bound in Proposition.~\ref{prop:LB} can be a thumb rule for choosing the best experimental implementation in terms of space overhead requirement. As the bounds in Propositions~\ref{prop:LB} and \ref{prop:threshold} extend to any qudit circuit, this can also be used as a thumb rule for comparison across different qudit and qubit technologies.

When the noise at the outputs of a gate are correlated, the minimum required space overhead and the noise threshold are decided by $I_c(\mc{N}^{(g)})$, where $\mc{N}^{(g)}$ is the generic non-i.i.d. noise at the outputs of a gate. In this context, it is important to understand the kind of correlations that hurt the computations most and plan to avoid such correlations in the experimental realizations.

Unlike \cite{Razborov2004,KempeRUW2008}, the bound in Proposition.~\ref{prop:LB} and that in \cite{FawziMS2022} require the knowledge of coherent information of a $k$-fold channel, $I_c(\mc{N}^{\otimes k})$. This often does not have a closed form and obtaining $I_c(\mc{N}^{\otimes k})$ requires to solve a $2^k$-dimensional non-convex optimization problem. With the increasing interest in understanding non-convex optimization in the machine learning community, a search for provably efficient algorithm for computing $I_c(\mc{N}^{\otimes k})$ can be of independent interest. 




\section{Proof of Proposition~\ref{prop:LB}}
\label{sec:proofLB}
For proving the lower bound in Proposition~\ref{prop:LB}, we first state a converse for  computing a class of functions $f$ on $d$-qubit states of dimension.

\begin{theorem}
\label{thm:lipschitzLB}
Suppose the function $f$ to be computed using the noisy quantum circuit satisfies the following conditions.

(i) $f$ has an inverse $f^{-1}$ that can be accurately computed if we have access to a noiseless quantum circuit, (ii) $f^{-1}$ exists and 
for any two $d$-qubit states $\eta_1^{(d)}$ and $\eta_2^{(d)}$, $1-F(f^{-1}(\eta_1^{(d)}), f^{-1}(\eta_2^{(d)})) \le L (1 - F(\eta_1^{(d)}, \eta_2^{(d)}))$ for some $L > 0$ and (iii) $L \le \frac{1-e^{-\frac{1}{8}}}{\epsilon}$.

Then, for $\epsilon$-accurate computation of $f$, for $d\ge 2g$, the required number of physical qubits is lower bounded by 
\begin{align}
\frac{d - \frac{1}{1-4\ln\frac{1}{1-\epsilon L}} I_c(\mc{N}^{\otimes g})}{\frac{I_c(\mc{N}^{\otimes g})}{g} + \frac{G}{g(G-1)} \frac{h_2(\frac{2g}{d} \ln\frac{1}{1-\epsilon L})}{1-4 \ln\frac{1}{1-\epsilon L}}}, \nonumber
\end{align}
\end{theorem}
{ where $G = \lceil \frac{N}{g} \rceil $, and $h_2(\cdot)$ is the binary entropy function.}
Proof of this theorem is presented in Appendix~\ref{sec:lipschitzLB}. Proof of Proposition~\ref{prop:LB} follows from this theorem.


\begin{proof}[Proof of Proposition~\ref{prop:LB}]
For proving Proposition~\ref{prop:LB}, we first consider the second term in the bound in Theorem~\ref{thm:lipschitzLB}:

\begin{align}
\frac{- \frac{1}{1-4\ln\frac{1}{1-\epsilon L}} I_c(\mc{N}^{\otimes g})}{\frac{I_c(\mc{N}^{\otimes g})}{g} + \frac{G}{g(G-1)} \frac{h_2(\frac{2g}{d} \ln\frac{1}{1-\epsilon L})}{1-4 \ln\frac{1}{1-\epsilon L}}}. \nonumber
\end{align}

Note that as $\frac{G}{g(G-1)} \frac{h_2(\frac{2g}{d} \ln\frac{1}{1-\epsilon L})}{1-\frac{4g}{d} \ln\frac{1}{1-\epsilon L}}>0$ { if $L \le \frac{1-e^{-{\frac{1}{8}}}}{\epsilon}$},
\begin{align}
 - \frac{1}{\frac{I_c(\mc{N}^{\otimes g})}{g} + \frac{G}{g(G-1)} \frac{h_2(\frac{2g}{d} \ln\frac{1}{1-\epsilon L})}{1-{4} \ln\frac{1}{1-\epsilon L}}}
\ge - \frac{g}{I_c(\mc{N}^{\otimes g})}.
\end{align}

On the other hand, $\frac{1}{1-4\ln\frac{1}{1-\epsilon L}} \le 2$ if $L \le \frac{1-e^{-{\frac{1}{8}}}}{\epsilon}$

Hence, the second term in the bound in Theorem~\ref{thm:lipschitzLB} is $\ge - 2~g$ if $L \le \frac{1-e^{-{\frac{1}{8}}}}{\epsilon}$.

Next, we consider the first term in the bound in Theorem~\ref{thm:lipschitzLB}:

\begin{align}
&~\frac{d}{\frac{I_c(\mc{N}^{\otimes g})}{g} + \frac{G}{g(G-1)} \frac{h_2(\frac{2g}{d} \ln\frac{1}{1-\epsilon L})}{1-4 \ln\frac{1}{1-\epsilon L}}} \nonumber,
\end{align}
and obtain a lower bound for this term.
Note that an upper bound on $h_2(\frac{2g}{d} \ln\frac{1}{1-\epsilon L})$ will give a lower bound on this term.

Next, we use the fact that 
\begin{align}  h_2(x) { \ln 2} & = x \ln\frac{1}{x} + (1-x) \ln\frac{1}{1-x} \nonumber \\
& \le x \ln\frac{1}{x} + \ln\frac{1}{1-x} \nonumber \\
& = x \ln\frac{1}{x} - \ln(1-x) \nonumber \\
& = x \ln\frac{1}{x} - (-x-\frac{x^2}{2}-\frac{x^3}{3}- \cdots) \nonumber \\
& \le  x \ln\frac{1}{x} + x (1+x+x^2+\cdots) \nonumber \\
& = x \left(\ln\frac{1}{x} + \frac{1}{1-x}\right), \nonumber
\end{align} 
 and under condition (iii) in Theorem~\ref{thm:lipschitzLB}, $h_2(\frac{2g}{d} \ln\frac{1}{1-\epsilon L}) \le h_2(\frac{g}{4d})$.

Thus, we obtain 
\begin{align} 
h_2(\frac{2g}{d} \ln\frac{1}{1-\epsilon L}) \le  \frac{g }{4d~{ \ln 2}} \left(\ln\frac{4d}{g}+\frac{1}{1-\frac{g}{4d}}\right) \le \frac{g}{4d~{ \ln 2}} \left(\ln\frac{4d}{g}+\frac{1}{1-\frac{1}{8}}\right) \label{eq:h2bound}
\end{align}
as $d \ge 2g$.
Under condition (iii) in Theorem~\ref{thm:lipschitzLB}, $\frac{1}{1-4 \ln\frac{1}{1-\epsilon L}}\le 2$ and hence,

\begin{align}
&~\frac{d}{\frac{I_c(\mc{N}^{\otimes g})}{g} + \frac{G}{g(G-1)} \frac{h_2(\frac{2g}{d} \ln\frac{1}{1-\epsilon L})}{1-4 \ln\frac{1}{1-\epsilon L}}} \nonumber \\
& \ge \frac{d}{\frac{I_c(\mc{N}^{\otimes g})}{g} + \frac{2G}{g(G-1)} h_2(\frac{2g}{d} \ln\frac{1}{1-\epsilon L})} \nonumber \\
& \ge \frac{d}{\frac{I_c(\mc{N}^{\otimes g})}{g}+\frac{2G}{g(G-1)}\frac{g}{4d~{\ln 2}} \left(\ln\frac{4d}{g}+\frac{8}{7}\right)} \nonumber \\
& \ge \frac{d}{\frac{I_c(\mc{N}^{\otimes g})}{g}+\frac{1}{2(d-g)~{\ln 2}}\left(\ln\frac{4d}{g}+\frac{8}{7}\right)}.\nonumber
\end{align}

The last inequality follows from the fact that $\frac{G}{G-1}$ is monotonically decreasing in $G$ and $G\ge\frac{d}{g}$.
%
%
%
%
%

If $f$ is unitary, then $f^{-1}$ satisfies all the three conditions in Theorem~\ref{thm:lipschitzLB} for $L=1$. Hence, the condition 
$0 < L \le \frac{1-e^{-{\frac{1}{8}}}}{\epsilon}$, is satisfied by a large class of $f$, including unitary transformations, if $0<\epsilon < 1-e^{-\frac{1}{8}} \le 0.11$. Thus, the derived bound is applicable to all practically important $f$,  including the well known quantum Fourier transform, which is the fundamental building block of many interesting algorithms.

This completes the proof of Proposition~\ref{prop:LB}.
\end{proof}



\section{Conclusion and Future Work}
\label{sec:conclusion}

In this paper, inspired by the information theoretic bounds for noisy classical circuits \cite{Pippenger1988, EvansS1999}, a connection between fault-tolerant quantum computation and finite blocklength communication is cultivated. This  leads to a lower bound on the required space overhead for fault-tolerant computation and is given by gate size $g$ divided by the coherent information of a $g$-fold noisy channel. This bound is tighter than the existing bounds and can be extended to the case where the noise on the outputs of a gate are correlated. It gives a tighter upper bound on the noise threshold. 

In future, we would like to combine our approach with techniques developed in \cite{Razborov2004,KempeRUW2008,FawziMS2022} for tightening the bound and design fault-tolerant schemes for achieving that bound. 

As discussed in Sec.~\ref{sec:result}, obtaining $I_c(\mc{N}^{\otimes k})$ involves a $2^k$-dimensional non-convex optimization problem. Since understanding non-convex optimization is of mathematical interest due to its usefulness in modern machine learning, an exploration for efficient computation of $I_c(\mc{N}^{\otimes k})$ can be of independent mathematical interest.

\section*{Acknowledgment}
{ The authors are thankful to the two anonymous reviewers and the editor,  Anurag Anshu, for their many helpful comments.} In particular, they thank Reviewer 1 for pointing out an error in the lower order term of the main bound. 

AC thanks Lav R. Varshney for nudging him multiple times to work on quantum circuits. AC also thanks Lav R. Varshney and Arul Lakshminarayan for useful feedback on the paper. AC gratefully acknowledges supports through the grants  NFIG IIT Madras, DST/INSPIRE/04/2016/001171 and SERB/SRG/2019/001809.

AN expresses gratitude for the support provided in part by National Science Foundation grant PHY-2112890.

\appendix

\section{Proof of Theorem~\ref{thm:lipschitzLB}}
\label{sec:lipschitzLB}
For proving \ref{thm:lipschitzLB}, we first state an important lemma.

\begin{lemma}
\label{lem:lipschitzLemma}
Suppose the conditions (i)--(iii) in Theorem~\ref{thm:lipschitzLB} hold. Then, for $\epsilon$-accurate computation we need
\(d \le \sum_{i=1}^G \frac{1}{1-2~\epsilon_i}\left(I_c(\mc{N}^{\otimes g}) + h_2(\epsilon_i)\right),\)
where, { $\{\epsilon_i\}$ are tunable auxiliary variables taking values in $[0,1]$ and satisfying the condition:} $\prod_{i=1}^G (1-\frac{\epsilon_i}{2}) \ge 1 - \epsilon~L$.
\end{lemma}

\begin{proof}[Proof of Theorem~\ref{thm:lipschitzLB}] Proof of this lemma is presented in Appendix~\ref{sec:lipschitzLemma}. Here, we prove Theorem~\ref{thm:lipschitzLB} using this lemma.

First, by using the fact that $1-x \le \exp(-x)$, we obtain the following upper bound.

\begin{align}
& \max_{\{\epsilon_i\}} \sum_{i=1}^G \frac{1}{1-2~\epsilon_i}\left(I_c(\mc{N}^{\otimes g}) + h_2(\epsilon_i)\right) \nonumber \\
& \mbox{s.t. } \sum_{i=1}^G \epsilon_i \le 2\ln\frac{1}{1-\epsilon L}. \label{eq:UB0}
\end{align}

Clearly the above upper bound is bounded by the sum of the maximum of the following optimization problems: P1 and P2.

\begin{align}
&\mbox{P1: } \max_{\{\epsilon_i\}} \sum_{i=1}^G \frac{1}{1-2~\epsilon_i} I_c(\mc{N}^{\otimes g})  \mbox{ s.t. } \sum_{i=1}^G \epsilon_i \le 2 \ln\frac{1}{1-\epsilon L}. \label{eq:P1}
\end{align}

\begin{align}
& \mbox{P2: } \max_{\{\epsilon_i\}} \sum_{i=1}^G  \frac{h_2(\epsilon_i)}{1-2~\epsilon_i} \mbox{ s.t. } \sum_{i=1}^G \epsilon_i \le 2 \ln\frac{1}{1-\epsilon L}. \label{eq:P2}
\end{align}

As P1 is a maximization of sum of convex functions subject to a linear constraint, P1 is maximized at an extreme point, as it is a convex maximization. Under condition (iii) in Theorem~\ref{thm:lipschitzLB}, $2 \ln\frac{1}{1-\epsilon L}<\frac{1}{2}$. Thus, the optimum in P1 is obtained when $\epsilon_1=2 \ln\frac{1}{1-\epsilon L}$ and $\epsilon_i=0$ for $i\ge 2$. 

Note that the optimum of P2 is upper-bounded by the optimum of the following problem.

\begin{align}
& \mbox{P3: } \max_{\{\epsilon_i\}} \frac{1}{1-4\ln\frac{1}{1-\epsilon L}} \sum_{i=1}^G  {h_2(\epsilon_i)} \mbox{ s.t. } \sum_{i=1}^G \epsilon_i \le 2 \ln\frac{1}{1-\epsilon L}. \label{eq:P3}
\end{align}

As P3 is a maximization of sum of concave functions subject to a linear constraint, by the symmetry of the problem, the optimum solution is $\epsilon^*_i=\frac{2}{G} \ln\frac{1}{1-\epsilon L}$. 

Summing the optimum of P1 and P3, the resulting upper bound on $d$ becomes

\begin{align}
 & d \le  (G-1) I_c(\mc{N}^{\otimes g}) + \frac{1}{1-4\ln\frac{1}{1-\epsilon L}} I_c(\mc{N}^{\otimes g}) + \frac{G}{1-{4}\ln\frac{1}{1-\epsilon L}} h_2(\frac{2}{G} \ln\frac{1}{1-\epsilon L}) \label{eq:P1P3} \\
 & ~ = (G-1) \left(I_c(\mc{N}^{\otimes g}) + \frac{G}{G-1} \frac{h_2(\frac{2}{G} \ln\frac{1}{1-\epsilon L})}{1-{4}\ln\frac{1}{1-\epsilon L}}\right) + \frac{1}{1-4\ln\frac{1}{1-\epsilon L}} I_c(\mc{N}^{\otimes g}). \nonumber
\end{align}

Note that as every physical qubit goes through a gate of size $g$ (including dummy identity gates as discussed in Sec.~\ref{sec:model}), $N \ge (G-1) g$ and hence,

\begin{align}
d \le N \left(\frac{I_c(\mc{N}^{\otimes g})}{g} + \frac{G}{g(G-1)} \frac{h_2(\frac{2}{G} \ln\frac{1}{1-\epsilon L})}{1-{4}\ln\frac{1}{1-\epsilon L}}\right) + \frac{1}{1-4\ln\frac{1}{1-\epsilon L}} I_c(\mc{N}^{\otimes g}) \nonumber
\end{align}

As all inputs must pass through a gate, $G \ge \frac{d}{g}$.  So, under condition (iii), $\frac{2}{G} \ln\frac{1}{1-\epsilon L}\le \frac{2g}{d} \ln\frac{1}{1-\epsilon L} \le 0.5$ {for any} $d \ge \frac{g}{2}$. As $h_2(x)$ is monotonic over $[0,0.5]$, 

\begin{align}
d \le N \left(\frac{I_c(\mc{N}^{\otimes g})}{g} + \frac{G}{g(G-1)} \frac{h_2(\frac{2g}{d} \ln\frac{1}{1-\epsilon L})}{1-{4}\ln\frac{1}{1-\epsilon L}}\right) + \frac{1}{1-4\ln\frac{1}{1-\epsilon L}} I_c(\mc{N}^{\otimes g}).
\end{align}

This implies

\begin{align}
N \ge \frac{d - \frac{1}{1-4\ln\frac{1}{1-\epsilon L}} I_c(\mc{N}^{\otimes g})}{\frac{I_c(\mc{N}^{\otimes g})}{g} + \frac{G}{g(G-1)} \frac{h_2(\frac{2g}{d} \ln\frac{1}{1-\epsilon L})}{1-{4}\ln\frac{1}{1-\epsilon L}}}.
\end{align}

This completes the proof of Theorem~\ref{thm:lipschitzLB}.
\end{proof}

\section{Proof of Lemma~\ref{lem:lipschitzLemma}}
\label{sec:lipschitzLemma}
For $f^{-1}$ satisfying conditions (i)-(iii) in Theorem~\ref{thm:lipschitzLB},  satisfaction of the $\epsilon$-accuracy condition C0 (Eq. \ref{eq:cri0}) implies satisfaction of the following condition.

\begin{align}
\mbox{C1: }\forall \rho^{(d)}~\mbox{(}d\mbox{-qubit state )}, F(f^{-1}\circ QC(\rho^{(d)}),\rho^{(d)}) > 1-\epsilon~L.
\label{eq:cri1}
\end{align}

Let us denote $\mc{D} \circ \mc{O}_D\circ \mc{O}_{D-1}\cdots \circ \mc{O}_2$ by $\mc{O}_{2:D}$, then
$QC(\rho^{(d)}) = \mc{O}_{2:D}\circ \mc{N}^{\otimes N} \circ \bar{\mc{O}}_1 (\rho^{(N)})$ and hence,
$f^{-1}\circ QC(\rho^{(d)})$ can be written as $\mc{H} \circ \mc{N}^{\otimes N} \circ \bar{\mc{O}}_1 (\rho^{(N)})$ for some $\mc{H}$. Here $\bar{\mc{O}}_1$ is the noiseless operation at layer $1$, i.e., $\mc{O}_1=\mc{N}^{\otimes N} \circ \bar{\mc{O}}_1$. Hence, satisfaction of condition C1 implies satisfaction of the following condition.

\begin{align}
 \exists \mc{H} \mbox{ s.t. } \forall \rho^{(d)}, F(\mc{H} \circ \mc{N}^{\otimes N} \circ \bar{\mc{O}}_1(\rho^{(N)}),\rho^{(d)}) > 1-\epsilon~L.
\label{eq:cri2}
\end{align}

The input to $\bar{\mc{O}}_1$ is $\rho^{(d)} \otimes \left(\otimes_{a=1}^{N-d}\sigma_a \right)$, which is a state of $N$ physical qubits, where $d$ of them are input  qubits (from $\rho^{(d)}$) and $N-d$ of them are ancilla qubits used for error correction. 

Let in an $\epsilon$-accurate implementation of $QC$, ${\mc{O}}_1$ is implemented using $G$ gates of size $g$ in parallel. In that implementation, let the $i$th gate on ${\mc{O}}_1$, for $1 \le  i\le G$, take $d_i$ qubits  from $\rho^{(d)}$ and rest from ancilla qubits as inputs. 

Without loss of generality, let us assume that in $\mc{O}_1$, the first $d_1$ qubits are input to gate $1$, next $d_2$ qubits are input to gate $2$ and so on. Let the noiseless computation by gate $i$ in layer $1$ be given by $\mc{G}_i(\cdot)$.
Then, for an input of the form $\rho^{(d)} = \otimes_{i=1}^G \rho_{d_i}$, where $\rho_{d_i}$ are $d_i$-qubit states, the output of $\mc{O}_1$ is

\(\mc{N}^{\otimes N} \circ \otimes_{i=1}^G \mc{G}_i\left(\rho_{d_i}\otimes \left(\otimes_{a=1}^{g-d_i}\sigma^{(i)}_{a} \right)\right),\)
where, $\sigma^{(i)}_{a}$ are the ancilla inputs to gate $i$.

So, for an $\epsilon$-accurate implementation, the following special case of Eq. \ref{eq:cri2} must be satisfied: $\exists \mc{H} \mbox{ s.t. for all } \rho^{(d)} = \otimes_{i=1}^G \rho_{d_i}$,
\begin{align}
  F\left(\mc{H} \circ \mc{N}^{\otimes N} \circ \otimes_{i=1}^G \mc{G}_i\left(\rho_{d_i}\otimes \left(\otimes_{a=1}^{g-d_i}\sigma^{(i)}_{a} \right)\right),\otimes_{i=1}^G \rho_{d_i}\right) > 1-\epsilon~L,
\label{eq:cri2a}
\end{align}

Since \(\mc{N}^{\otimes N} \circ \otimes_{i=1}^G \mc{G}_i(\rho_{d_i}\otimes \left(\otimes_{a=1}^{g-d_i}\sigma^{(i)}_{a} \right) = \otimes_{i=1}^G \mc{N}^{\otimes g} \circ \mc{G}_i(\rho_{d_i}\otimes \left(\otimes_{a=1}^{g-d_i}\sigma^{(i)}_{a} \right),\)
and $\rho^{(d)} = \otimes_{i=1}^G \rho_{d_i}$, if there exists an $\mc{H}$ satisfying the condition in 
Eq. \ref{eq:cri2a}, then there also exists an $\mc{H}=\otimes_{i=1}^G \mc{H}_i$ satisfying the same condition. This is because $\mc{G}_i\left(\rho_{d_i}\otimes \left(\otimes_{a=1}^{g-d_i}\sigma^{(i)}_{a} \right)\right)$ is independent of  $\mc{G}_j(\cdot)$ and $\rho_{d_j}$ for $j\neq i$ and noise are independent across gates.

Thus, a necessary condition for satisfaction of C0 is: $\exists \{\mc{H}_i\} \mbox{ s.t. for all } \rho^{(d)} = \otimes_{i=1}^G \rho_{d_i}$,
\begin{align}
 F\left(\otimes_{i=1}^G \mc{H}_i \circ \mc{N}^{\otimes g} \circ  \mc{G}_i(\rho_{d_i}\otimes \left(\otimes_{a=1}^{g-d_i}\sigma^{(i)}_{a} \right),\otimes_{i=1}^G \rho_{d_i}\right) > 1-\epsilon~L.
\label{eq:cri3}
\end{align}

By the fact that $F(\sigma_1 \otimes \sigma_2, \sigma'_1 \otimes \sigma'_2) = F(\sigma_1,\sigma'_1) \cdot  F(\sigma_2, \sigma'_2)$,  condition \eqref{eq:cri3} is equivalent to: for some $\{\epsilon_i \in (0,1): i\}$, for each $i$, 
\begin{align}
\mbox{C2: } \exists \mc{H}_i \mbox{ s.t. } \forall \rho_{d_i}, F(\mc{H}_i \circ \mc{N}^{\otimes g} \circ  \mc{G}_i(\rho_{d_i}\otimes \left(\otimes_{a=1}^{g-d_i}\sigma^{(i)}_{a} \right)), \rho_{d_i}) > 1-\epsilon_i,
\end{align}
where, $\prod_{i=1}^G (1-\epsilon_i) \ge 1 - \epsilon~L$.

For satisfaction of condition C2, a necessary condition is the following:  for some $\{\epsilon_i \in (0,1): i\}$, for each $i$, 
\begin{align}
\mbox{C3: } \exists \{\mc{H}_i, \bar{\mc{G}}_i\} \mbox{ s.t. } \forall \rho_{d_i}, F(\mc{H}_i \circ \mc{N}^{\otimes g} \circ  \bar{\mc{G}}_i(\rho_{d_i}\otimes \left(\otimes_{a=1}^{g-d_i}\sigma^{(i)}_{a} \right)), \rho_{d_i}) > 1-\epsilon_i,
\end{align}
where,  $\prod_{i=1}^G (1-\epsilon_i) \ge 1 - \epsilon~L$. Note that under condition (iii) in Theorem~\ref{thm:lipschitzLB}, $1-\epsilon L > \frac{1}{2}$ and hence, a necessary condition is $\epsilon_i < \frac{1}{2}$ for $1\le i \le G$.

Note that condition C3 is equivalent to the $\epsilon_i$-accuracy criteria for one-shot communication of quantum information over the quantum channel $\mc{N}^{\otimes g}$ \cite{KhatriW2020book}. Hence, by 
\cite[Sec.~9.1.2]{KhatriW2020book}, a necessary condition for satisfaction of C3 is: for $1 \le i \le G$
\begin{align} 
d_i \le \frac{1}{1-2~\epsilon_i}\left(I_c(\mc{N}^{\otimes g}) + h_2(\epsilon_i)\right),
\end{align}
where $h_2$ is the binary entropy function.


From this, the following upper bound on $d$ for $\epsilon$-accurate computation is obtained: for some $\{\epsilon_i\}$,

\(d=\sum_{i=1}^G d_i \le \sum_{i=1}^G \frac{1}{1-2~\epsilon_i}\left(I_c(\mc{N}^{\otimes g}) + h_2(\epsilon_i)\right),\)
where, $\prod_{i=1}^G (1-\epsilon_i) \ge 1 - \epsilon~L$.

By allowing noiseless forward classical communication, for the same rate of quantum communication, fidelity of $1-\eta$ can at most improve to $1-\frac{\eta}{2}$ \cite[Sec.~VIII]{BarnumKN2000}. Using this fact, if follows that by allowing noiseless classical computation and classical buffers, the bound can at most be the following.

\(d=\sum_{i=1}^G d_i \le \sum_{i=1}^G \frac{1}{1-2~\epsilon_i}\left(I_c(\mc{N}^{\otimes g}) + h_2(\epsilon_i)\right),\)
where, $\prod_{i=1}^G (1-\frac{\epsilon_i}{2}) \ge 1 - \epsilon~L$. This completes the proof of Lemma~\ref{lem:lipschitzLemma}.

The same proof goes through after replacing $\mc{N}^{\otimes g}$ by a general $2^g$-dimensional channel $\mc{N}^{(g)}$. Hence, the bound applies if the noise at the outputs of a gate are not independent, but noise are independent across gates.

\section{Proof of Proposition~\ref{prop:threshold}}
\label{sec:thresholdProof}
Consider the bound \eqref{eq:P1P3} obtained in the proof of Theorem~\ref{thm:lipschitzLB}. When $I_c(\mc{N}^{\otimes g})=0$ and $\epsilon<0.11$ this bounds become

\begin{align}
d \le 2G \cdot h_2(\frac{2}{G} \ln\frac{1}{1-\epsilon L}). \nonumber
\end{align} 

Then, using the facts that $h_2(x)$ is monotonic over $[0,0.5]$ and ${{h_2(x) \ln 2}}\le x\left(\ln\frac{1}{x}+\frac{1}{1-x}\right)$, we obtain the following bound.

\begin{align}
d & \le 2G \cdot \frac{2}{G{\ln 2}} \ln\frac{1}{1-\epsilon L} \left(\ln\frac{G}{2\ln\frac{1}{1-\epsilon L}} +\frac{G}{G - 2\ln\frac{1}{1-\epsilon L}} \right) \nonumber \\
&  \le  \frac{4}{{\ln 2}}\ln\frac{1}{1-\epsilon L} (\ln G + \frac{4}{3}) - \frac{4}{{\ln 2}} \ln\frac{1}{1-\epsilon L} \ln\left(2\ln\frac{1}{1-\epsilon L}\right)  \nonumber \\ 
& \le \frac{1}{2{\ln 2}} (\ln G + \frac{4}{3}) + \frac{1}{2{\ln 2}} \ln 4 . \nonumber
\end{align} 
The first inequality follows because $2\ln\frac{1}{1-\epsilon L} \le \frac{1}{4}$ under conditions (i)-(iii) in Theorem~\ref{thm:lipschitzLB}, $G\ge 1$ and $\frac{G}{G-\frac{1}{4}}$ is decreasing with $G$. 

For the last inequality, { the second term is maximized when $L = \frac{1-e^{-1/8}}{\epsilon}$, which implies that $- 4 \ln\frac{1}{1-\epsilon L} \ln\left(2\ln\frac{1}{1-\epsilon L}\right) \leq \frac{1}{2} \ln 4$.}


As $G = \lceil \frac{N}{g} \rceil \leq \frac{N}{g}+1$, the bound implies

\begin{align}
N \ge  \frac{g}{4} ~\exp(2 ~{\ln 2} ~d-\frac{4}{3})-1.
\end{align}


This completes the proof of Proposition~\ref{prop:threshold}.

\section{Towards a Stronger Impossibility Result}
\label{sec:sandwichedThreshold}
Using the same proof as that of Lemma~\ref{lem:lipschitzLemma} and the results in \cite[Sec.~9.1.2]{KhatriW2020book}, one can obtain another bound on $d$, in terms of $\alpha$-sandwiched Renyi coherent information $I_c(\mc{N}^{\otimes g};\alpha)$ of the channel. Here, $\alpha $ is a parameter more than $1$.

\(d=\sum_{i=1}^G d_i \le \sum_{i=1}^G \left(I_c(\mc{N}^{\otimes g};\alpha) + \frac{\alpha}{\alpha-1}{\log_2}\frac{1}{1-\epsilon_i}\right),\)
where, $\prod_{i=1}^G (1-\frac{\epsilon_i}{2}) \ge 1 - \epsilon~L$. 


Therefore, the bound becomes
\begin{align}
G I_c(\mc{N}^{\otimes g};\alpha) + \frac{\alpha}{\alpha-1} \max_{\{\epsilon_i\}}  \sum_{i=1}^G {\log_2}\frac{1}{1-\epsilon_i} \nonumber \\
\mbox{ s.t. } \prod_{i=1}^G (1-\frac{\epsilon_i}{2}) \ge 1 - \epsilon~L.
\end{align}

Since $1-x \le \exp(-x)$ and loosening the constraint increases the maximum, the maximum of the above optimization is upper-bounded by
\begin{align}
G I_c(\mc{N}^{\otimes g};\alpha) + \frac{\alpha}{{ \ln 2 (\alpha-1)}} \max_{\{\epsilon_i\}}  \sum_{i=1}^G \ln\frac{1}{1-\epsilon_i} \nonumber \\
\mbox{ s.t. } \sum_{i=1}^G \epsilon_i \le 2\ln\frac{1}{1-\epsilon L}.
\end{align}

Note that \(-\ln(1-x) = x + \frac{x^2}{2} + \frac{x^3}{3} + \cdots \le x + x^2 + x^3 + \cdots = \frac{x}{1-x},\)
and hence, for getting an upper-bound, it is enough to maximize $ \sum_{i=1}^G \frac{\epsilon_i}{1-\epsilon_i}$.

As $\frac{x}{1-x}$ is convex on $[0,1)$, the maximization is obtained at an extreme point, i.e., $\epsilon_1=2\ln\frac{1}{1-\epsilon L}$ and $\epsilon_i=0$ for $i\ge 2$. Thus, the bound becomes
\( G I_c(\mc{N}^{\otimes g};\alpha) + \frac{\alpha}{{ \ln 2 (\alpha-1)}} \frac{2\ln\frac{1}{1-\epsilon L}}{1-2\ln\frac{1}{1-\epsilon L}}.\)

Under the conditions in Theorem~\ref{thm:lipschitzLB}, $\frac{2\ln\frac{1}{1-\epsilon L}}{1-2\ln\frac{1}{1-\epsilon L}}\le \frac{1}{3}$ and hence, the upper bound on $d$ is $\frac{\alpha}{{ 3 \ln 2 }(\alpha-1)}$, whenever $I_c(\mc{N}^{\otimes g};\alpha)=0$.

Thus, for noise strengths that result into $I_c(\mc{N}^{\otimes g};\alpha)=0$, fault-tolerant computation is impossible for \(d \ge \frac{\alpha}{{ 3 \ln 2 }(\alpha-1)},\)
i.e., there is an upper bound on the number of (logical) qubits that can be computed with $\epsilon$-accuracy.

It is known that $I_c(\mc{N}^{\otimes g};\alpha) \downarrow I_c(\mc{N}^{\otimes g})$ as $\alpha \downarrow 1$ \cite{KhatriW2020book} and hence, it can be possible to take the impossibility threshold arbitrarily close to $I_c(\mc{N}^{\otimes g})=0$.

The fact $I_c(\mc{N}^{\otimes g};\alpha) \downarrow I_c(\mc{N}^{\otimes g})$ as $\alpha \downarrow 1$ implies for $\alpha\le \alpha_{th}$, for some $\alpha_{th}>1$,  
\(\frac{I_c(\mc{N}^{\otimes g})}{g} \le \frac{I_c(\mc{N}^{\otimes g};\alpha)}{g} \le \sup_{k \ge 1} \frac{I_c(\mc{N}^{\otimes k};\alpha)}{k}.\)

This implies that no computation with number of qubits $d > \frac{\alpha_{th}}{{ 3 \ln 2 } (\alpha_{th}-1)}$ is possible for certain noise strengths, which include the noise strengths for which the quantum capacity is zero, i.e., $\sup_{k \ge 1} \frac{I_c(\mc{N}^{\otimes k};\alpha)}{k}$ $=0$.

%
%

\bibliographystyle{plain} 
\bibliography{NoisyQuantum}
\end{document}